\documentclass[12pt]{article}
\usepackage[a4paper, total={7in, 9.5in}]{geometry}

\usepackage{geometry,url,textcomp,mathtools,array,amsthm,amssymb,hyperref,xcolor,color,natbib}
\usepackage{cleveref}

 \newtheorem{theorem}{Theorem}[section]
\newtheorem{example}{Example}%
\newtheorem{remark}{Remark}%
\newcommand{\assign}{:=}
\newcommand{\R}{\mathbb{R}}

\newtheorem{definition}{Definition}%

\newcommand{\ba}{\begin{eqnarray}}
\newcommand{\ea}{\end{eqnarray}}

\makeatletter
\newcommand\footnoteref[1]{\protected@xdef\@thefnmark{\ref{#1}}\@footnotemark}
\makeatother

\providecommand{\keywords}[1]
{
  \small	
  \textbf{\textit{Keywords---}} #1
}

\begin{document}

\title{Rough volatility: fact or artefact?}

\author{Rama Cont \\Mathematical Institute, University of Oxford\and Purba Das\footnote{Email: purbadas@umich.edu}\\ Department of Mathematics, University of Michigan
\footnote{We thank Peter Friz, James-Michael Leahy, Mathieu Rosenbaum and Alexander Schied for helpful comments and discussions.}}

\maketitle

\begin{abstract}
We investigate the statistical evidence for the use of `rough' fractional processes with Hurst exponent $H< 0.5$ for  modeling the volatility of financial assets, using a model-free approach.
We introduce a non-parametric method for estimating the roughness of a function based on discrete sample, using the concept of normalized $p$-th variation along a  sequence of partitions. Detailed numerical experiments based on sample paths of fractional Brownian motion and other fractional processes reveal good finite sample performance of our estimator for measuring the roughness of sample paths of stochastic processes.   We then apply this method to estimate the roughness of realized volatility signals based on high-frequency observations. Detailed numerical experiments based on  stochastic volatility models show that, even when the instantaneous volatility has diffusive dynamics with the same roughness as Brownian motion, the realized volatility exhibits rough behaviour corresponding to a Hurst exponent significantly smaller than $0.5$.
Comparison of roughness estimates  for realized and instantaneous volatility in  fractional volatility models with different values of Hurst exponent  shows that, irrespective of the roughness of the spot volatility process, realized volatility always exhibits `rough' behaviour with an apparent Hurst index $\widehat{H}<0.5$.
These results suggest that  the origin of the roughness observed in realized volatility time-series lies in the estimation error rather than the volatility process itself.
\end{abstract}




\newpage
\tableofcontents
\newpage
\section{Introduction}
\subsection{Fractional processes in finance: from long-range dependence to `rough volatility'}

Beginning with  \cite{mandelbrot1968}, fractional Brownian motion and fractional Gaussian noise have been used as building blocks of stochastic models of various phenomena in physics, engineering \citep{levy2005} and finance \citep{baillie1996,bollerslev1996,comte1998,cont2007,gatheral2014,Rogers1997,willinger1999}. 
Fractional Brownian motion has two remarkable properties which have contributed towards its adoption as a building block in stochastic models: first, its ability to model long-range dependence, as measured by the slow decay $\sim T^{2H-2}$ of auto-correlation functions of increments, where $0<H<1$ is the Hurst exponent;  second, its ability to generate trajectories which have varying levels of H\"older regularity (`roughness').
The former is a property that manifests itself over long time scales while the latter manifests itself over short time scales and, in general, these two properties are unrelated.
But in the case of fractional Brownian motion, the two properties are linked through self-similarity and governed by the Hurst exponent $0<H<1$: for $H>1/2$ one obtains long-range dependence in increments and trajectories smoother than Brownian motion while for $H<1/2$ one obtains `anti-correlated' increments and trajectories rougher than Brownian motion\footnote{Incidentally, Hurst and H\"older happen to have the same initials, adding to the confusion...}. Processes driven by fractional Gaussian noise with $H<1/2$ are thus sometimes referred to as `rough processes'.

In early applications to financial data (\cite{baillie1996,bollerslev1996,comte1998,willinger1999}), fractional processes were adopted in order to model {\it long range dependence} effects in financial time series \cite{cont2005}. More specifically, statistical evidence of {\it volatility clustering} \cite{cont2007} - positive dependence of the amplitude of returns over long time scales - led to the development of  stochastic volatility models driven by fractional Brownian motion. A well-known example of such a fractional stochastic volatility model is the one proposed by  \cite{comte1998} who  modelled the dynamics of the (instantaneous) volatility $\sigma(t)$ of an asset as:
\begin{equation}
Y(t)=\ln \sigma(t)\qquad dY(t)=-\gamma Y(t) dt + \theta dB^H(t)\label{eq.comterenault}
\end{equation}
where $B^H$ is a fractional Brownian motion (fBM) with Hurst exponent $H$.
This long-range dependence in volatility is modelled by choosing   $1> H>1/2$ (\cite{comte1998,bollerslev1996,breidt1998,hurvich2005,lahiri2020}).

A recent strand of literature, starting with \cite{gatheral2014}, has suggested the use of fractional Brownian models with $H<1/2$ for modelling volatility.  Unlike previous studies based on auto-correlations of various volatility estimators over long time scales \citep{baillie1996,bollerslev1996,cont2005}, \cite{gatheral2014} rely on the analysis of the behaviour of volatility estimators over short intraday time scales in order to assess the `roughness' of these signals and  concluded that volatility is `rough' i.e. has  paths with a H\"older regularity which is strictly less than $1/2$, suggest to use stochastic models with sample paths rougher than Brownian motion.

However, it has not been lost on experts working in this area that previous estimation results for fractional models in the  literature on long-range dependence in volatility,  pointed  towards Hurst exponents $H> 0.5$  (and around $0.55$) \citep{baillie1996,comte1998,lahiri2020}while the recent `rough volatility' literature indicates Hurst exponents much smaller than $0.5$ and closer to $0.1$.
Together with the well-known statistical issues plaguing the estimation of Hurst exponents \citep{beran,roger2019}, these conflicting results call for a critical examination of the empirical evidence for `rough volatility'.


Compounding this issue is the fact that (spot) volatility is not directly observed but estimated from price series, with an inherent estimation error which has been the subject of  many studies \citep{barndorff2002,jacod2011,lahiri2020}. This estimation error is far from i.i.d.: it is known to  possess path-dependent features (see \cite{jacod2011}). As a result, measures of roughness for realized volatility indicators may be quite different from those of the underlying `spot volatility'. This is simply because the convergence of high-frequency volatility estimators in $L^p$ norms does not imply in any way their functional convergence in H\"older norms or other norms related to roughness.

As already pointed out by  \cite{Rogers1997}, these two properties, namely the short-range behaviour which determines the roughness of the path, and the long-range dependence property, can (and should) be modeled through different mechanisms.  \cite{pakkanen2021} discuss several such approaches.

The focus of the literature on parametric models based on fractional Brownian motion or fractional Gaussian noise concentrates these two, very different, properties in a single parameter: the Hurst exponent $H$ \citep{pakkanen2022,fukasawa2019,gatheral2014}. 
Such parametric approaches proceed as follows: one estimates a parametric model for volatility dynamics based on some fractional Gaussian driving noise with Hurst exponent $0< H<1$ using a MLE \citep{fukasawa2019} or method of moments \citep{pakkanen2022}. Then, based on the estimated value of this parameter $H$, one concludes that ``volatility is rough" if $\hat{H}< 0.5$.

The validity of such approaches hinges   on the assumption that the class of models used is well-specified. As pointed out by \cite{pakkanen2021}, this is unlikely to be the case for SDEs driven by fractional Gaussian noise if one wants to accommodate both (long-range) dependence properties and (short-range) roughness  properties.

To avoid this caveat, we propose a model-free {\it non-parametric} method which focuses solely on the roughness properties of sample paths.  Although less ambitious in its scope -we only focus on roughness properties rather than developing a full model for volatility dynamics- our approach is robust to the specification errors and estimation biases which plague parametric methods.

\subsection{Contribution}
We address these questions in detail by re-examining the statistical evidence from high-frequency financial data  in an attempt to clarify whether the assertion that `volatility is rough' (i.e. rougher than typical paths of Brownian motion) is supported  by empirical evidence. We investigate the statistical evidence for the use of `rough' fractional processes with Hurst exponent $H< 0.5$ for the modelling of volatility of financial assets, using a non-parametric, model-free approach.

We introduce a non-parametric method for estimating the roughness of a function/path based on a (high-frequency) discrete sample, using the concept of normalized $p$-th variation along a sequence of partitions, and discuss the consistency of our estimator in a pathwise setting. We investigate the finite sample performance of our estimator for measuring the roughness of sample paths of stochastic processes using detailed numerical experiments based on sample paths of fractional Brownian motion and other fractional processes. We then apply this method to estimate the roughness of realized volatility signals based on high-frequency observations. Through a detailed numerical experiment based on a  stochastic volatility model, we show that even when the instantaneous (spot) volatility has diffusive dynamics with the same roughness as Brownian motion, the realized volatility exhibits rough behaviour corresponding to a Hurst exponent significantly smaller than $0.5$. Similar behavior is observed in financial data as well, which suggests that the origin of the roughness observed in realized volatility time-series lie in the estimation error rather than the volatility process itself.
Comparison of roughness estimates  for realized and instantaneous volatility in fractional volatility models for different values of Hurst parameter $H$ shows that whatever the value of $H$ for the (spot) volatility process, realized volatility always exhibits `rough' behaviour.  

Our results are broadly consistent with the observations by \cite{roger2019}, but we pinpoint more precisely the origin of the apparent `rough' behaviour of volatility as being the estimation error inherent in the estimation of realized volatility (sometimes  known as microstructure noise). In particular, our results question whether the empirical evidence presented from high-frequency volatility estimates supports the `rough volatility' hypothesis.


\section{Measuring the roughness  of a path}

Determining the roughness of realized volatility from a sample path plays a crucial role in model specification \cite{gatheral2014,fukasawa2019}. In practice, we observe only a single price path so one is faced with the problem of determining the roughness of a process from a single price path sampled at high frequency.  We present in this section several concepts for measuring the roughness of a path and discuss how they may be used to design estimators from high-frequency observations.

\subsection{{$p$}-th variation and roughness index of a path}
Consider a sequence of partitions $\pi=(\pi^n)_{n\geq 1}$ of $[0,T]$ where  
\ba \pi^n=\left(0=t^{n}_0<t^n_1<\cdots<t^{n}_{N(\pi^n)}=T\right) \ea
represents observation times `at frequency $n$'.
We denote $N(\pi^n)$ to be the number of intervals in the partition $\pi^n$. Denote respectively by
$| \pi^n | =\sup_{i=1,\cdots,N(\pi^n)}| t^n_i-t^n_{i-1}|,$ and $ \underline{\pi^n }=\inf_{i=1,\cdots,N(\pi^n)}| t^n_i-t^n_{i-1}|, $ 
the size of the largest and the smallest interval of $\pi^n$. In this paper, we will always assume
$$| \pi^n | =\sup_{i=1,\cdots,N(\pi^n)}| t^n_i-t^n_{i-1}| \mathop{\to}^{n\to\infty} 0. $$



The concept of $p$-th variation along a sequence of partitions  $\pi=\left(\pi^n\right)_{n\geq 1}$ with $0=t_0^n<...< t^n_k<...< t^n_{N(\pi^n)}=T$ is defined following \cite{perkowski2019}:

\begin{definition}($p$-th variation along a sequence of partitions)\label{def:p-var}
 $x \in C^0([0,T],\R)$  has finite $p$-th variation along the sequence of partitions $\pi=(\pi^n, n\geq 1)$  if there exists a continuous increasing function $[x]^{(p)}_\pi:[0,T]\to\mathbb{R}_+$ such that
	\begin{equation}
	\forall t\in[0,T],\qquad \sum_{\substack{[t^n_j, t^n_{j + 1}] \in \pi^n: \\ t^n_j \le t}} \left| x(t^n_{j + 1}) - x(t^n_j)\right| ^p\mathop{\longrightarrow}^{n\to\infty} [x]^{(p)}_\pi(t).\label{eq.pointwisecv}\end{equation}
If this property holds, then the convergence in \eqref{eq.pointwisecv} is  uniform. We call $[x]^{(p)}_\pi$ the \emph{$p$-th variation} of  $x$ along the sequence of partitions $\pi$. We denote $V^p_\pi([0,T],\mathbb{R})$ the set of all continuous paths with finite $p$-th variation along $\pi$.
\end{definition}

 \begin{remark} For $x\in V^p_\pi([0,T],\mathbb{R})$ we have $[x]^{(p)}_\pi(T)<\infty$. If $$\sum_{\pi^n\cap(a,b)}\left|x(t^n_{i+1})-x(t^n_{i})\right|^p \to \infty$$ for all $(a,b)\subset[0,T]$ then we will write $[x]^{(p)}(t)  = \infty$.
\end{remark}

To formalize the concept of roughness, we introduce the notion of variation index and roughness index of a path:
\begin{definition}[Variation index]\label{def.variationindex}
The variation index of a path $x$ along a partition sequence $\pi$ \ is defined as the smallest $p\geq 1$ for which $x$ has finite $p$-th variation along $\pi$:
    \[p^\pi(x) =  \inf\left\{p\geq 1 \;:\; x\in V^p_\pi([0,T],\mathbb{R}) 
    \right\} .\]
\end{definition}
 
\begin{definition}[Roughness index]\label{def.roughnessindex}
The roughness index of a path $x$ (along $\pi$) is defined as 
\[{H}^\pi(x) = \frac{1}{p^\pi(x)}.\]
\end{definition}
When the underlying sequence of partitions is clear, we will omit $\pi$ and denote these indices as  $p(x)$ and $ {H}(x)$. A similar roughness index was introduced by \cite{Schied2022}.

For a (real-valued) stochastic process $X:[0,T]\times \Omega \to \mathbb{R}$ the roughness index $p^\pi(X(.,\omega))$ of each sample path $X(.,\omega)$ may be  different  in principle. Nevertheless there are many important classes of stochastic processes which have an {\it almost-sure} roughness index.
For example, the roughness index of fractional Brownian motion (fBM) matches with the corresponding Hurst parameter/ H\"older exponent:
\begin{example}
Brownian motion $B$ has variation index $p^\pi(B)=2$   and roughness index $H^\pi(B)=\frac{1}{2}$ along any refining partition sequence $\pi$ or any partition $\pi$ with $|\pi^n|\log{n}\to 0$ \cite{dudley1973}.

\par Fractional Brownian motion $B^H$ has variation index $p^\mathbb{T}(B^H)=\frac{1}{H}$ and roughness index $H^\mathbb{T}=H$ along the dyadic partition sequence $\mathbb{T}$.
\end{example}
In general, the existence of a variation index is not obvious. For further details see \cite{das2022theory}

\subsection{Normalized  $p$-th variation}
It is not easy to use $p$-th variation directly on empirical data for estimating roughness based on discrete observations, as this involves checking convergence to an unknown limit. 
We introduce a normalized version of $p$-th variation which has better asymptotic properties \cite{das2022theory}:
\begin{definition}[Normalized $p$-th variation along a sequence of partitions]
Let $\pi$ be a sequence of partitions of $[0,T]$ with mesh $|\pi^n| \to 0$ and  $\pi^n = \left(0=t^n_1<t^n_2<\cdots< t^n_{N(\pi^n)}=T\right)$.  $x\in V^p_{\pi}([0,T],\mathbb{R})$ is said to have {\sl normalized $p$-th variation along} $\pi$ if  there exists a continuous function $w(x,p,\pi):[0,T]\to \mathbb{R}$ such that:
\ba\label{eq.weighted.p.var} \forall t\in[0,T], \quad \sum_{\pi^n\cap [0,t]}\frac{\left| x\left(t^n_{i+1}\right)-x\left(t^n_{i}\right)\right|^p}{[x]^{(p)}_{\pi}(t^n_{i+1})-[x]^{(p)}_{\pi}(t^n_{i})} \times \left(t^n_{i+1} -t^n_{i}\right) \xrightarrow[]{n\to \infty} w(x, p, \pi)(t).\ea
We denote $ N^p_\pi([0,T],\mathbb{R})$ the class of all continuous functions for which the normalized $p$-th variation\footnote{For $p=2$ we will call this quantity as `normalized quadratic variation'.} exists.
\end{definition}
The terminology is justified by the following result \cite{das2022theory} which shows that, for a large class of functions with $p-$th variation, the  normalized $p$-th variation   exists and is linear:
\begin{theorem} \label{linear.weighted.variation}
Let $x\in V^{p}_\pi([0,T],\mathbb{R})$ for some $p>1$ where  $\pi$ be a sequence of partitions of $[0,T]$ with vanishing mesh $|\pi^n|\to 0$. Furthermore, if the $p$-th variation  is absolutely continuous  then:\[x\in N^{p}_\pi([0,T],\mathbb{R}) \qquad \text{ and  } \qquad \forall t\in [0,T],\; w(x,p,\pi)(t)=t. \] 
\end{theorem}
\begin{proof}
See Appendix.
\end{proof}
The following result shows that normalized $p$-th variation is a `sharp' statistic: if a function has finite $p$-th variation along a sequence of partitions $\pi$  then for all $q\neq p$ the normalized $p$-th variation is either infinite or zero.
\begin{theorem}\label{infinite.zero.weighted.variation}
Let $\pi$ be a sequence of partitions of $[0,T]$ with mesh $|\pi^n|\to 0$. Let $x\in  V^{p}_\pi([0,T],\mathbb{R})$ with $[x]^{(p)}_\pi \in (0,\infty)$ for some $p>1$. 
\begin{itemize}
    \item[(i)] For all $t\in (0,T]$ and for all $ q>p$; $w(x,q,\pi)(t)=\infty$.
    \item[ (ii)]  For all $t\in [0,T]$ and for all $ q<p$; $w(x,q,\pi)(t)=0$
\end{itemize}
\end{theorem}
\begin{proof}
The proof is provided in Appendix.
\end{proof}
In particular, Brownian motion almost surely has linear normalized quadratic variation.
\begin{example}[Normalized quadratic variation for Brownian motion]\label{weightedQV-BM}
Let $B$ be a Wiener process on a probability space $(\Omega, {\cal F},\mathbb{P})$, and
$(\pi^n)_{n\geq 1}$ be a sequence of partitions of $[0,T]$ with  $|\pi^n|\log n \to 0$. Then:
\[\mathbb{P}\left(w(B,2,\pi)(t) = t \right)=1. \]
\end{example}
\begin{proof}
The proof is provided in the Appendix.
\end{proof}
\begin{remark}
In the above example, instead of taking any partition sequence with $|\pi^n|\log n \to 0$, we can also take any refining sequence of partitions. The proof is then different and uses the martingale convergence theorem. 
\end{remark}
\begin{example}[Stochastic integrals]\label{stochastic.int} Let $X(t)=\int_0^t \sigma(u)dB_u$  where $\sigma$ is an adapted process with  $ \int_0^T \sigma^2(u)du < \infty$. Then, for any refining partition sequence $\pi$  with vanishing mesh,
\[\mathbb{P}\left(w(X,2,\pi)(t)=t \right)=1. \]
\end{example}
\begin{proof}
This is an immediate consequence of Theorem \ref{linear.weighted.variation}.
\end{proof}
\begin{remark}  
In the statement of Example \ref{stochastic.int}, we can replace the assumption of refining partitions with partitions satisfying $|\pi^n|\log(n)\to 0$.
\end{remark}

\begin{example}[Normalized $p$-th variation for fractional Brownian motion]\label{weightedQV-fBM}
Let $B^H$ be a fractional Brownian motion with Hurst index $H$, on a probability space $(\Omega, {\cal F},\mathbb{P})$. $B^H$ has normalized p-th variation along the dyadic partition $\mathbb{T}=(\mathbb{T}^n)_{n\geq 1}$ for $p=1/H$ almost-surely:  $$\mathbb{P}\left(  w\left(B^H,\frac{1}{H},\mathbb{T}\right)(t)= t \right)=1.$$
\end{example}
\begin{proof}
The proof is a consequence of \cite[Proposition 4.1]{viitasaari2019}.
\end{proof}



\subsection{Estimating roughness from discrete observations}

Given observations on a refining time partition $\pi^L$, we  define the `normalized $p$-th variation statistic' which is the discrete counterpart of  the normalized p-th variation:
\ba \label{statistic.W} W(L,K,\pi,p,t,X)\assign \sum_{\pi^K\cap [0,t]}\frac{\left|X(t^K_{i+1})-X(t^K_{i})\right|^p}{\sum_{\pi^L\cap [t^K_i,t^K_{i+1}]}\left|X(t^L_{j+1})-X(t^L_{j})\right|^p} \times \left(t^K_{i+1}-t^K_{i}\right).\quad\ea
The definition of the statistic \eqref{statistic.W} involves two frequencies: a `block' frequency $K$ and a sampling frequency $L\gg K$.
As the partition is refining, $\pi^K$ is a subpartition of $\pi^L$.  
The denominator is estimated by grouping the sample of size $L$ into $K$ many groups, where each group contains $\frac{L}{K}$  consecutive data points.

The statistic \eqref{statistic.W} converges to the  normalized $p$-th variation \eqref{eq.weighted.p.var}
as the sampling frequency $L$ and block frequency increase:
$$\mathop{\lim}_{K\to\infty}\mathop{\lim}_{L\to\infty}W(L,K,\pi,p,t,x) = w(x,p,\pi)(t).$$
It is thus natural to define roughness estimators for a discretely sampled signal in terms of \eqref{statistic.W}.

The {\it variation index estimator} $\widehat{p}_{L,K}(X)$  associated with the signal sampled on $\pi^L$ is then obtained by computing $W(L,K,\pi,p,t,X)$ for different values of $p$ and  solving the following equation for $p^\pi_{L,K}(X)$,
\begin{equation}
  W(L,K,\pi, \widehat{p}^\pi_{L,K}(X),T,X) = T.   \label{eq.statistic}
\end{equation}
One can either fix a window length $T$ or solve  \eqref{eq.statistic} in a least squares sense across several values of $T$.

An estimator for the   roughness index   is subsequently defined as:
\begin{equation} \widehat{H}^\pi_{L,K}(X) = \frac{1}{\widehat{p}^\pi_{L,K}(X)}.\label{eq.RoughnessEstimator}
\end{equation}
We will denote the roughness estimator \eqref{eq.RoughnessEstimator} as $\widehat {H}_{L, K}$ when the underline dataset and the corresponding partition sequence are clear. Asymptotic properties of these estimators under high-frequency sampling are studied in \cite{das2022theory}.
\subsection{Finite sample behaviour of the roughness estimator} \label{sec.simulation}
We will now study the finite sample behaviour of the roughness estimator $\widehat{H}^\pi_{L, K}(X)$ using high-frequency simulations of fractional Brownian motions.
In the simulation examples  unless mentioned otherwise we will use  a uniform partition sequence of $[0,1]$ with:
\[\pi^n=\left(0<\frac{1}{n}<\frac{2}{n}<\cdots<1 \right).\] 
To assess the finite sample accuracy of the estimator we compare the roughness index estimator $\widehat{ H}^\pi_{L,K}$ with the underlying Hurst exponent $H\in \{0.1,0.3,0.5,0.8\}$. For every simulated path we compute  $W(L=300\times300,K=300,\pi,q,t=1,X=B^H)$ for different values of $q$, in order to estimate $\widehat{H}_{L,K}$. In figure \ref{fig-fbm.sim}, the black line is the value of $\log\big(W(L=300\times300,K=300,\pi,q=p,t=1,X=B^H)\big)$ plotted against roughness index $1/p$ in log-scale. The blue horizontal line represents the estimated roughness index $\widehat{ H}_{L, K}$ whereas the dotted horizontal line represents the  Hurst parameter. 
Figure \ref{fig-hist1} shows the  histograms of the estimator $\widehat{H}^\pi_{L, K}$  generated from $150$ independent paths. We   observe that for   datasets with length  $L=300\times300$ our roughness estimator $\widehat H_{L=300\times300, K=300}$   has satisfactory accuracy. Table \ref{table.1}  provides summary statistics for roughness index $\widehat H$ of simulated fractional Brownian motions. 
\begin{table}[b!]
    \centering
   
\[\begin{tabular}{|c|c|c|c|c|c|c|}  \hline 
   H&    Min. & Lower quartile & Median  &  Mean & Upper quartile &    Max. \\ \hline
0.1 & 0.0650 & 0.0920 & 0.1030 & 0.1009 & 0.1100 & 0.1440   \\ \hline
0.3 &  0.2730 & 0.2940 & 0.2980 & 0.2976 & 0.3020 & 0.3180 \\ \hline
0.5 &  0.4820 & 0.4940 & 0.4980 & 0.4978 & 0.5020 & 0.5140 \\ \hline
0.8 &  0.7570 & 0.7820 & 0.7900 & 0.7891 & 0.7940 & 0.8220 \\ \hline
\end{tabular}\]
   \caption{Summary statistics for estimated roughness index $\widehat{H}_{L, K}$ for fractional Brownian motion $B^H$ with $L=300\times 300, K=300$.}
   \label{table.1}
\end{table}
Figure \ref{fig-high.freq.hist} represents a similar plot for simulated fractional Brownian motion with Hurst parameter $H=0.1$.  In Figure \ref{fig-high.freq.hist}, in left, similar to Figure \ref{fig-fbm.sim}, $\log(W(L=2000\times2000,K=2000,\pi,p,t=1,X=B^H))$ is plotted against $H=1/p$ and the right plot represents the  histograms  of the estimator  $\widehat{H}_{L=2000\times 2000, K=2000}$. The summary statistics for the esimator are provided in Table \ref{table.2}.
\begin{table}[b!]
    \centering
   \[\begin{tabular}{|c|c|c|c|c|c|c|c|}  \hline
      Hurst index H & Min. & Lower quartile & Median  &  Mean & Upper quartile &    Max. \\ \hline
0.1& 0.086& 0.096& 0.099& 0.099& 0.103& 0.117  \\ \hline
\end{tabular}\] \caption{Summary statistics for estimated roughness index $\widehat{H}_{L, K}$ for fractional Brownian motion $B^H$ with $H=0.1, L=2000\times 2000, K=2000$.}
   \label{table.2}
\end{table}
 To compute the estimator $\widehat{H}_{L,K}$   we have different possible choices of  $K\ll L$. Figure \ref{fig-diff.K} shows how the estimator  $\widehat{H}^\pi_{L, K}$ varies with $K$ for  fractional Brownian motion with Hurst parameter $H=0.1$. The black line represents the $\widehat{H}_{L, K}$ plotted against different values of $K$ whereas the blue vertical line represents the value for $ L=300\times300,K=300$. We observe that  when we vary  $K$ in the neighbourhood $K\approx \sqrt{L}$ the estimator performs reasonably well and is not very sensitive to the choice of $K$ in this vicinity.  

In summary, we observe that for realistic sample sizes and frequencies encountered in high-frequency financial data, the estimator is quite accurate and not sensitive to the block size $K$ in the range $K\approx \sqrt{L}$.

\begin{figure}[ht!]
    \centering
    \includegraphics[width=1 \textwidth]{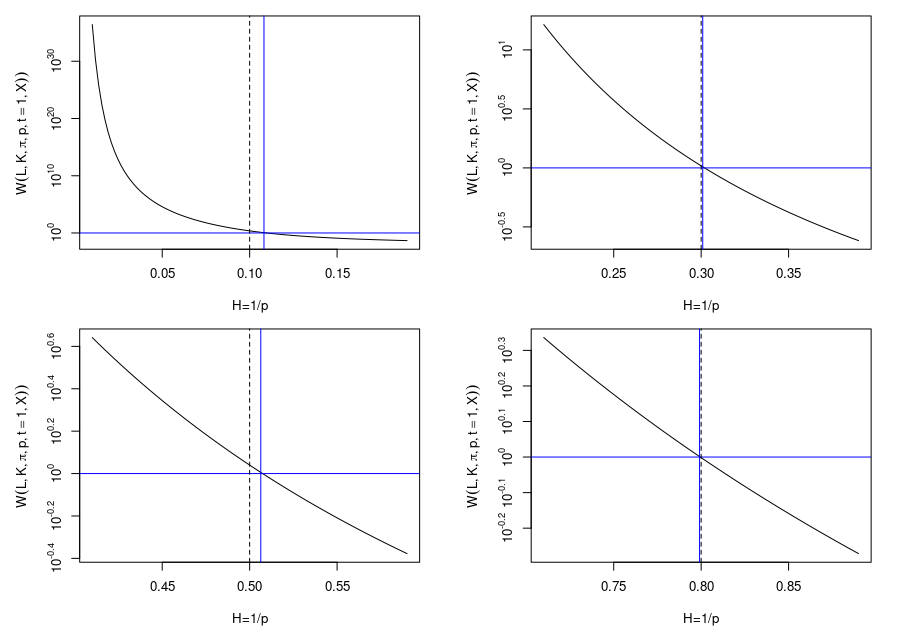}
    \caption{Log scale plot of the normalized p-variation statistic for  fractional Brownian motion with Hurst parameter $H\in \{0.1,0.3,0.5,0.8\}$. The black solid line represents the log of normalized $p$-th variation statistics plotted against $H=1/p$. The blue vertical line represents $\widehat H_{L, K}$ using the normalized $p$-th variation statistics (with $L=300\times 300, K=300)$. The horizontal dotted line represents the true Hurst parameter $H$.
}
    \label{fig-fbm.sim}
\end{figure}

\begin{figure}[ht!]
    \centering
    \includegraphics[width=1 \textwidth]{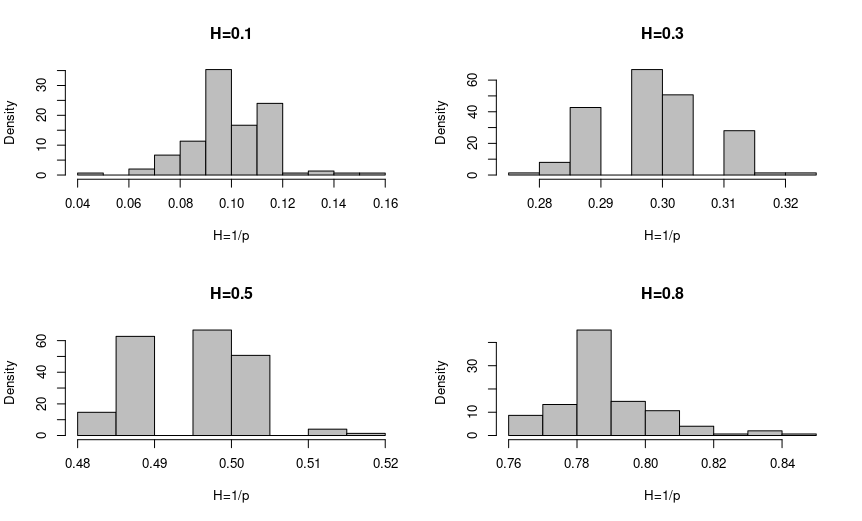}
    \caption{Histogram of estimated roughness index $\widehat H_{L,K}$ with $ L=300\times300,K=300$ generated from $150$ independent simulations of fractional Brownian motion with Hurst parameter $H\in \{0.1,0.3,0.5,0.8\}$.}
    \label{fig-hist1}
\end{figure}

\begin{figure}[ht!]
    \includegraphics[width=1 \textwidth]{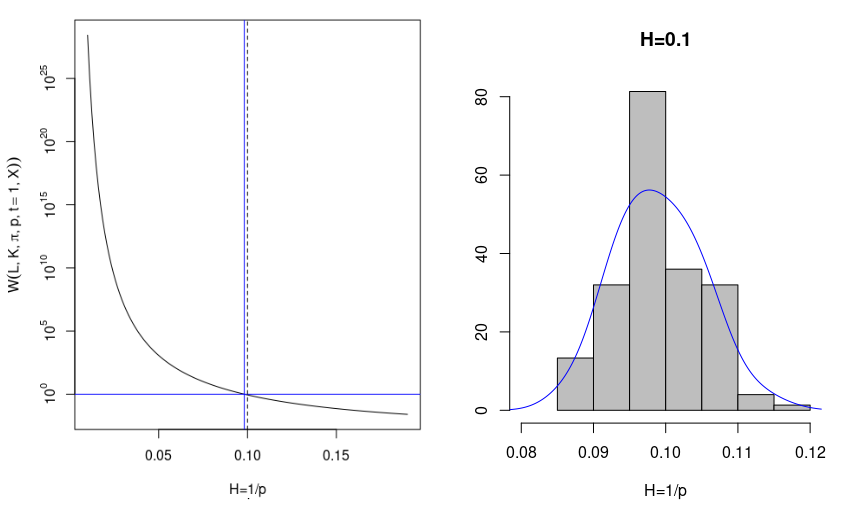}
    \caption{Simulation results for fractional Brownian motion with Hurst parameter $H=0.1$. \textbf{Left: } The log of normalized $p$-th variation statistic is plotted against $H=1/p$ in black. The blue vertical line represents the estimated roughness index $\widehat H_{L, K}$ (with $L=2000\times2000,K=2000$), whereas the dotted line represents the true value $H=0.1$. \textbf{Right: } Histogram of estimated roughness index $\widehat H_{L, K}$ generated by simulating $n=150$ independent fractional Brownian paths with Hurst parameter $H=0.1$. The blue line represents a kernel estimator for the density.}
    \label{fig-high.freq.hist}
\end{figure}

\clearpage
\begin{figure}[ht!]
    \centering
    \includegraphics[width=1 \textwidth]{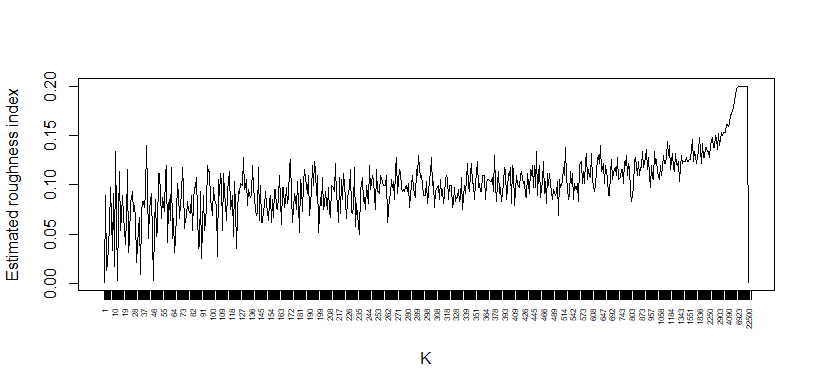}
    \caption{The solid black line represents the estimated roughness index $\widehat H_{L=300\times300, K}$ 
    plotted against different values of $K$ for a  fractional Brownian motion with Hurst parameter $H=0.1$.}
    \label{fig-diff.K}
\end{figure}

\section{Spot volatility and realized volatility}
Contrary to prices of an asset which may be observed and sampled directly from market data, {\it (spot) volatility} is not directly observable and as a consequence must be estimated from prices. Stochastic volatility models represent the price of a financial asset as the solution of a stochastic differential equation driven by a Brownian motion: 
 \begin{equation}
     dS_t= \sigma_t S_t dB_t + \mu_t S_t dt \label{eq.SV}
 \end{equation}
where the coefficient $\sigma_t$ represents the instantaneous or `spot' volatility.
In general, $\sigma_t$ is represented as a random process itself driven by fractional processes.

In a practical situation, the price  $S_t$ at time $t$, is usually observed over a non-uniform time grids of $[0,T]$: 
\ba  \pi^n=\left(0=t^{n}_0<t^n_1<\cdots<t^{n}_{N(\pi^n)}=T\right).\ea

In order to study high-frequency asymptotics of roughness estimators, we assume $|\pi^n|\to 0$ as $n $ increases;  here the index $n$ may be thought of as a `sampling frequency'. An example to keep in mind is the dyadic partition sequence: $\pi^n=\left(t^n_i= \frac{iT}{2^n}, i=0,\cdots,2^n\right)$ but none of the results below requires a uniform grid.

The spot volatility process $\sigma_t$  may then be recovered from the {\it quadratic variation} of the log-price $X= \log S$ along this particular grid:
\begin{equation}
    \label{eq.spotvol} \sigma^2_t= \frac{d}{dt} [\log S ]_\pi(t)
\end{equation} where the quadratic variation $[\log S]_\pi$ of the log-price
 $$[\log S ]_\pi(t) =\mathop{\lim}_{n\to\infty}\sum_{\pi^n\cap [0,t]} \left(\log \frac{S(t^n_{i+1})}{ S(t^n_i)}\right)^2   =\lim_{n\to\infty}RV_t(\pi^n)^2$$
is computed as a high-frequency limit of the {\it realized variance} \citep{barndorff2002,andersen2003} along the sampling grid $\pi^n$, defined as
\ba\label{eq.sigma.hat1} RV_t(\pi^n)^2 = \sum_{\pi^n\cap [0,t]} \left(\log \frac{S(t^n_{i+1})}{ S(t^n_i)}\right)^2 = \sum_{\pi^n\cap [0,t]} \left(X(t^n_{i+1})-X(t^n_i)\right)^2.\ea
The realized volatility  is defined as the square root of the realized variance.

\begin{definition}[Realized volatility]
The {\it realized volatility} of a price process $S$ over time interval $[t,t+\delta]$ sampled along the time partition $\pi^n$ is defined as:
\ba\label{eq.sigma.hat} RV_{t,t+\Delta}(\pi^n)= \sqrt{\sum_{\pi^n\cap [t,t+\Delta]} \left(X(t^n_{i+1})-X(t^n_i)\right)^2}=\sqrt{[X]_{\pi^n}(t+\Delta)-[X]_{\pi^n}(t)} \ea
where $X= \log S$.
\end{definition}

If the price $S_t$ follows a stochastic volatility model \eqref{eq.SV} with {\bf instantaneous volatility} $\sigma_t$ then realized variance converges to the quadratic variation of $\log S$ (also called `{\it integrated variance}') as sampling frequency increases \citep{barndorff2002,jacod2011}:
\ba  \quad \label{eq-microstruction} RV_{t}(\pi^n)^2 \mathop{\to}^{\mathbb{P}}_{n\to\infty} IV_{t} :=\int_{0}^{t} \sigma^2_u du,\qquad RV_{t,t+\Delta}(\pi^n) \mathop{\to}^{\mathbb{P}}_{n\to\infty} \sqrt{IV_{t,t+\Delta} }=\sqrt{\int_{t}^{t+\Delta} \sigma^2_u du}.\ea
Along a single price path observed at high-frequency, we can compute the realized variance \eqref{eq.sigma.hat} and the realized volatility $RV_{t,t+\delta}(\pi^n)$ in \eqref{eq-microstruction} may be used as a practical indicator of volatility:
$$ RV_{t,t+\Delta}(\pi^n) \simeq \sqrt{\Delta}\ \sigma_t.$$

Several empirical studies have attempted to estimate the roughness of `realized volatility' signals using high-frequency data, i.e. \cite{andersen2003,barndorff2002,contmancini2011,jacod2011,podolskij2009}. A well-known reference is the study of \cite{gatheral2014} where the authors estimate the roughness index of S\&P500 realized volatility by performing the following logarithmic regression to :
\ba m(q,\Delta):=  \mathit{\frac{1}{n}[\log RV]^{(q)}_{\pi^n} }= \frac{1}{n}\sum_{t=1}^{n} |\log(RV_{t+\Delta})-\log(RV_t)|^q \approx  C_q \Delta^{\xi_q}.
\label{m(q,delta)}\ea 
The coefficients $\xi_q$ are also shown to behave linearly in $q$:
$$\xi_q \approx q\  \widehat{H}_R.$$ Regression of $\xi_q$ on $q$ yields an estimate $\widehat{H}_R$ of Hurst/H\"older index, for which \cite{gatheral2014} report the value  $\widehat{H}_R=0.13$.  
Based on these observations, they propose a fractional SDE for (spot) volatility:
$$ d\log\sigma_t^2 = \mu_t dt  + \eta dB^H_t.$$
As we see from Equation \ref{m(q,delta)}, the  method used in \cite{gantert1994} actually uses $p$-th variation of the $\log (RV)$ to calculate the roughness of the underlying volatility process. 
Figure \ref{fig-gatheral-original} is a replication of the log-regression model described above to estimate the roughness index of the volatility of $5$-min  S\&P 500. 
However, in an interesting simulation study  using paths simulated from a Brownian OU volatility process,  \cite{roger2019} showed that the scaling behavior claimed as evidence for `rough volatility' is also observed in  a Brownian OU model over a range of time scales, and 
that estimators of the roughness index  based on log-regression of empirical $p$-th variation have poor accuracy. 
\par Similar evidence for the lack of accuracy of such estimators based on log-regression of $p$-th variation  is  shown by \cite{fukasawa2019}.

\begin{figure}[ht!]
    \centering
\includegraphics[width=1 \textwidth]{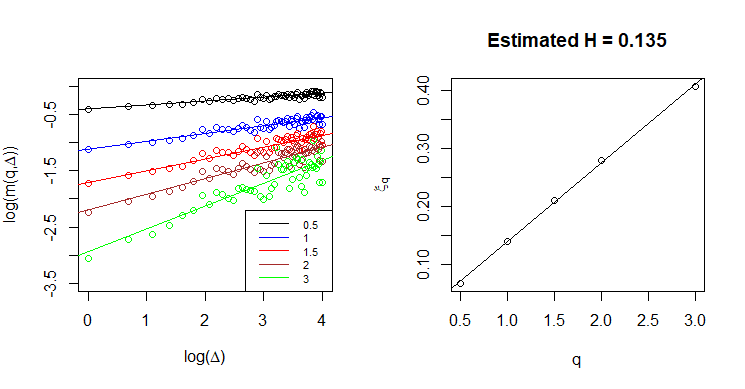}
    \caption{\textbf{Left: }   A reproduction of the log-regression method introduced by  \cite{gatheral2014} using SPX 5-minute realized volatility from the Oxford-Man Institute’s Realized Library. \textbf{Right: } 
    Regression of $\xi_q$ vs $q$: the  estimated slope is $\widehat{H}_R = 0.135$.}\label{fig-gatheral-original}
\end{figure}


\section{Numerical experiments}
We now compare various estimators of the roughness index  for instantaneous volatility $\sigma_t$ with those obtained from   realized volatility $RV_{t,t+\Delta}(\pi^n)$ using price trajectories simulated from stochastic volatility models with varying degrees of ``roughness".

\subsection{Stochastic volatility diffusion models}\label{sec.SV}

Let us first consider the following stochastic volatility where volatility is simply (the modulus of) a Brownian motion:
\begin{example}\label{example.simulation}
\ba\label{model} dS_t=\sigma_tS_tdB_t, \;\text{ with}\quad \sigma_t=|W_t|,\ea
where  $B, W$ are two  independent Brownian motions.
\end{example}
Figure \ref{fig-realized.vol.simulation} represents a path of the instantaneous volatility $\sigma_t$ and the realized volatility $RV_t(\pi^L)$  computed as in Equation \ref{eq.sigma.hat}  by taking 300 consecutive data-points, which corresponds to a 5-minute moving window. The right plot of Figure \ref{fig-realized.vol.simulation} represents the estimation error, which is defined as the difference of instantaneous and realized volatility. The ACF of the estimation error shows a complex time-dependent pattern which rules out IID behavior and indicates a complex dependence structure.

The estimated roughness index of instantaneous and realized volatility are observed to be very different. In the left graph of Figure \ref{fig-comp.H.simulatred} we plot $\log(W(K=500,L=500\times500,\pi,p,t=1,X=RV))$ against $H=1/p$ for the realized volatility. On the other hand,  the right graph is the same plot with the same set of parameters but for instantaneous volatility. The estimated roughness index for realized volatility ($\widehat H_{L=500\times500, K=500}(RV)=0.27$) is significantly smaller than the roughness index of the instantaneous volatility ($\widehat H_{L=500\times500, K=500}(\sigma)=0.49$) suggesting rougher behaviour of realized volatility.  
 As in our simulation study we do not have any measurement errors, this roughness behaviour purely comes from estimation error. In some studies it is assumed that the estimation error or the log-estimation error is I.I.D. (see for example \cite{fukasawa2019}) but as one can see from this diffusion example, both the estimation error and the log-estimation error is far from I.I.D. and hence the assumptions put forth for example in \cite{fukasawa2019}, is not very realistic for general diffusion models.

\par  The solid black lines in Figure \ref{fig-diff.K.realized} and Figure \ref{fig-diff.K.realized.instentanouus} respectively represent the estimated roughness index $\widehat H_{L= 300\times300,K}$ plotted against different values of $K$ for the realized and instantaneous volatility (model \ref{model}). The blue vertical line represents for $K=500, L=500\times500$. From the figures, we can observe that irrespective of the choice of $K$ for the finite sample dataset of length $L=500\times 500$, the realized volatility is significantly rougher than the instantaneous volatility.

We now compare our roughness estimator with  the log-regression method suggested in \cite{gatheral2014} for the model \ref{model}. It turns out that even with the log-regression model, similar `rougher' realized volatility is observed even if the instantaneous volatility has Brownian diffusive behaviour. Figure \ref{fig-gatharal-trueVol} and Figure \ref{fig-gatheral-realizedVol} show that the realized volatility has a significant smaller roughness index than the instantaneous volatility even with respect to the log regression method. In this example it is clear that the roughness observed in realized volatility is attributable to the discretization error (`estimation error') and not the roughness of the spot volatility process, which is Brownian.
 
 \begin{figure}[ht!]
    \centering
    \includegraphics[width=1 \textwidth]{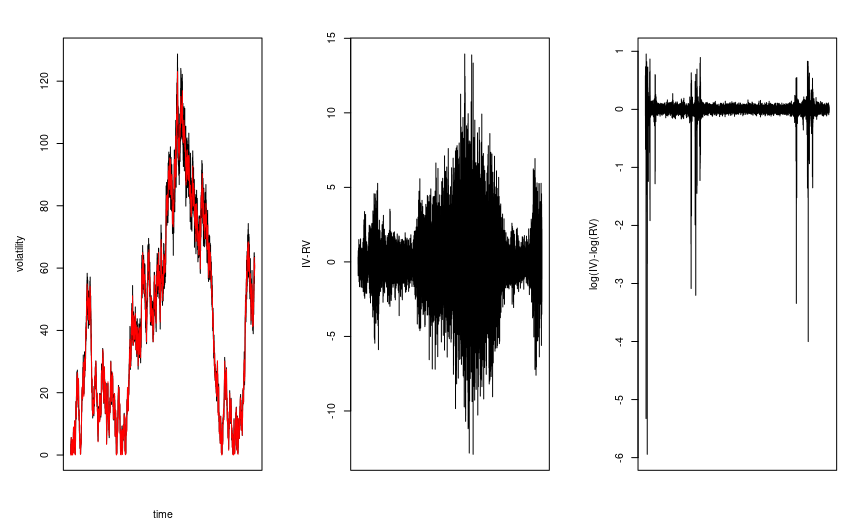}
    \caption{Simulation model: $\sigma_t=|B'_t| , dS_t=S_t\sigma_tdB_t$, where  $B_t$ and $B'_t$ are Brownian motions independent of each other. \textbf{Left:} The red line represents instantaneous volatility $\sigma_t$ whereas the black line represents realized volatility $RV_t$. \textbf{Right:} Corresponding estimation error for the left simulated path. } 
    \label{fig-realized.vol.simulation}
    \end{figure}

\begin{figure}[ht!]
\centering
    \includegraphics[width=1 \textwidth]{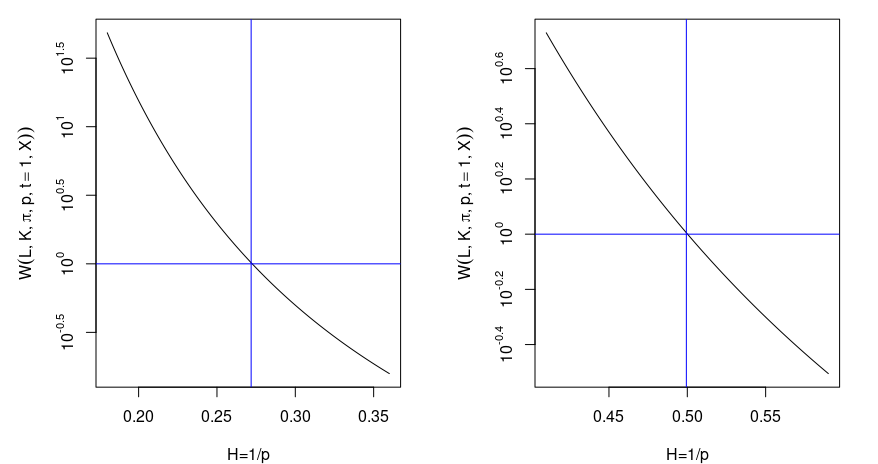}
    \caption{\textbf{Left: } Estimated roughness index $\widehat H_{L,K}$ (via normalized $p$-variation statistic with $L=500\times 500,K=500$), for realized volatility derived from a Brownian diffusion model, shown in Figure \ref{fig-realized.vol.simulation}. \textbf{Right: } Estimated roughness index $\widehat H_{L,K}$ for instantaneous volatility of the same price path. }
    \label{fig-comp.H.simulatred}
\end{figure}

\begin{figure}[ht!]
    \includegraphics[width=.9 \textwidth]{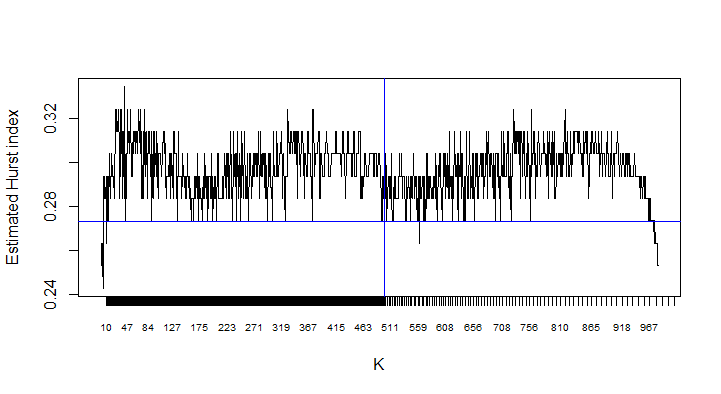}    
    \caption{The solid black line represents the estimated roughness index $\widehat H_{L,K}$ via normalized $p$-th variation statistic $W(L= 500\times500,K,\pi,q,t= 1,X=RV)$ plotted against different values of $K$ for the realized volatility shown in Figure \ref{fig-realized.vol.simulation}. The blue vertical line represents $L=500\times 500,K=500$ whereas the blue horizontal line represents $\widehat H=0.273$.}
    \label{fig-diff.K.realized}
    \end{figure}

\begin{figure}[ht!]
    \includegraphics[width=.9 \textwidth]{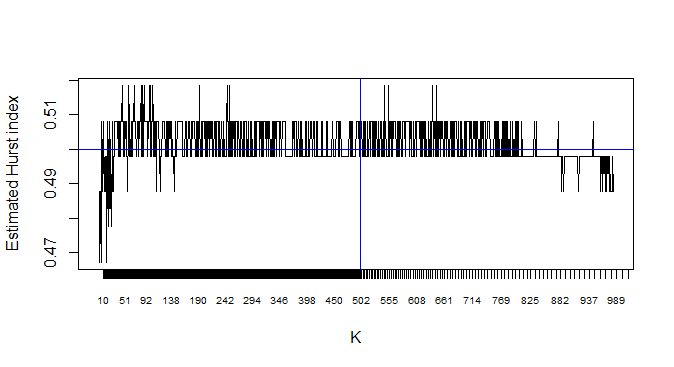}
    \caption{The solid black line represents the estimated roughness index $\widehat H_{L,K}$ via normalized $p$-th variation statistic $W(L= 500\times500,K,\pi,q,t= 1,X=IV)$ plotted against different values of $K$ for the instantaneous volatility shown in Figure \ref{fig-realized.vol.simulation}. The blue vertical line represents $ L=500\times500,K=500$ whereas the blue horizontal line represents true Hurst parameter $H=0.5$.}
    \label{fig-diff.K.realized.instentanouus}
\end{figure}

\begin{figure}[ht!]
\includegraphics[width=.9 \textwidth]{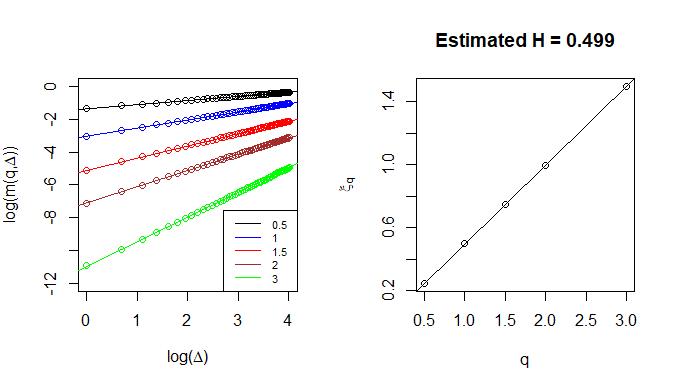}
        \caption{\textbf{Left: }  Scaling analysis of instantaneous volatility of a  simulated Brownian stochastic volatility model (Example \ref{example.simulation}) using the same log-regression method introduced in \cite{gatheral2014}. \textbf{Right: } regression coefficients $\xi_q$ as a function of $q$. The estimated roughness index is $\widehat{H}_R = 0.499$.  }\label{fig-gatharal-trueVol}
\end{figure}

\begin{figure}[ht!]
    \centering
\includegraphics[width=.9 \textwidth]{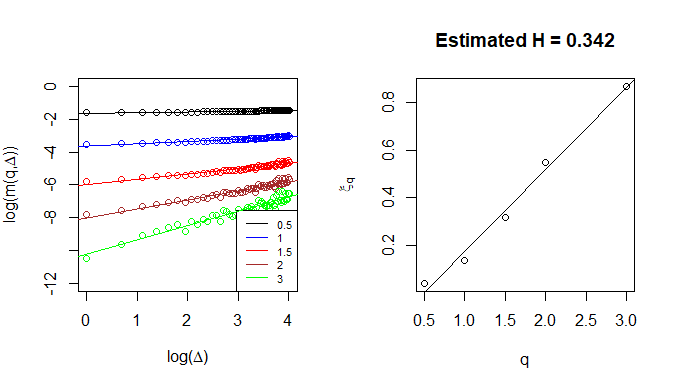}
    \caption{\textbf{Left: }  Scaling analysis of realized volatility of a  simulated Brownian stochastic volatility model (Example \ref{example.simulation}) using the same log-regression method introduced in  \cite{gatheral2014}. \textbf{Right: } regression coefficients $\xi_q$ as a function of $q$. The estimated roughness index is $\widehat{H}_R = 0.342$.}\label{fig-gatheral-realizedVol}

\end{figure}

Next we consider a more realistic mean-reverting volatility model in which the volatility follows a Brownian Ornstein–Uhlenbeck process:
\begin{example}[OU-SV model]\label{ex-OU-SV}
\ba dS_t=S_t\sigma_tdB_t ,\qquad
\sigma_t=\sigma_0 e^{Y_t}, \quad dY_t=-\gamma Y_tdt+\theta dB_t' \label{eq.OU-SV}\ea
where $B_t$ and $B'_t$ are two independent Brownian motions.
\end{example}
In the simulation, we use $\sigma_0=1, Y_0=0$ and $\gamma=\theta=1$.
The left plot of Figure \ref{fig-microstructure.ou} represents the realized volatility (respectively instantaneous volatility) of the price process in black (respectively red) simulated from the above stochastic volatility model \ref{eq.OU-SV}. The middle plot of Figure \ref{fig-microstructure.ou} represents the corresponding estimation error, which is the difference between the realized and the instantaneous volatility. Visually the middle plot suggests the estimation error has a complicated dependence structure. But unlike the plot for Example \ref{example.simulation}, the log error on the right plot of Figure \ref{fig-microstructure.ou} has an I.I.D. Gaussian structure (This is supported by the theory provided in \cite{fukasawa2019}). 
\par Now we compare the distribution of the estimator $\widehat{H}_{L,K}$ with $(L=300\times 300,K=300)$ for realized and instantaneous volatility across $2500$ independent paths drawn from \eqref{eq.OU-SV}. 
The left plot in Figure \ref{fig-kernel.plot} is the distribution of $\widehat{H}_{L,K}$ for the realized volatility while the right plot corresponds to the same for instantaneous volatility.  The following table provides summary statistics for the   estimator $\widehat{H}_{L,K}$ with ${L=300\times 300,K=300}  $ across  $2500$ independent sample paths  for realized volatility  and instantaneous volatility respectively. 
\begin{table}[t]
    \centering
\[\begin{tabular}{| c| c| c |} \hline
  & Realized volatility & Instantaneous volatility \\ \hline
 Min.  & 0.087     & 0.528  \\ \hline 
 1st Quantile &    0.128    &0.552 \\ \hline 
 Median    &     0.136    & 0.556 \\ \hline 
 Mean  &    0.137    &   0.557  \\ \hline 
 3rd Quantile    &   0.148    &   0.563 \\ \hline 
 Max.   &  0.181    &    0.581  \\ \hline

\end{tabular}\]
    \caption{Estimated roughness index $\widehat{H}_{L,K}, L=300\times300,K=300  $ for realized volatility  and instantaneous volatility for the diffusion model \eqref{model} with $H=0.5$. }
    \label{table.3}
\end{table}

\begin{figure}[ht!]
        \centering
        \includegraphics[width=0.9 \textwidth]{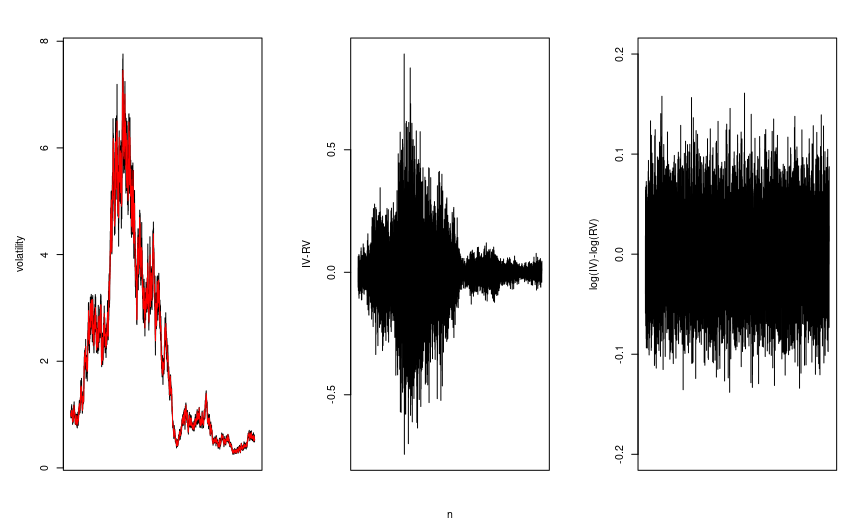}
        \caption{\textbf{Left: } Realized volatility (in black) vs instantaneous volatility (red) for the OU price model   \eqref{eq.OU-SV}. \textbf{Right: }  Estimation error.}
        \label{fig-microstructure.ou}
    \end{figure}

\begin{figure}[ht!]
    \centering
    \includegraphics[width=.9 \textwidth]{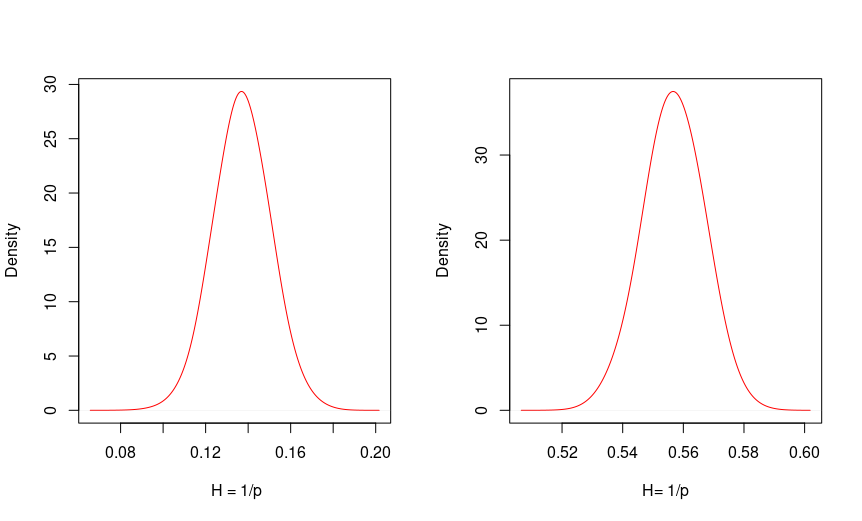}
    \caption{Distribution of the estimated roughness index $\widehat{H}_{L,K}$ for $(L=300\times 300,K=300)$ across  $2500$ independent simulations for the OU-SV model \eqref{eq.OU-SV} with roughness index $H=0.5$. \textbf{Left: } Realized volatility. \textbf{Right: } Instantaneous volatility. }
    \label{fig-kernel.plot}
\end{figure}
\clearpage

\subsection{A fractional Ornstein-Uhlenbeck model}
In both  previous examples, instantaneous volatility follows a diffusive behaviour similar to Brownian motion with $H=\frac{1}{2}$,  yet the realized volatility exhibits  "rough" behaviour with an estimated roughness index significantly smaller than $0.5$. 
\par We now consider the more general case of a price process generated by a stochastic volatility model of the type \eqref{eq.comterenault} where instantaneous volatility has a general roughness index $H\in (0,1)$ and explore how the roughness index of the instantaneous volatility reflects on the roughness index of realized volatility and the corresponding estimation error.
\begin{example}\label{ex-fou}[Fractional OU process]
Consider the following price process, where the volatility is described by a fractional Ornstein–Uhlenbeck process:
\ba dS_t= \sigma_tS_t dB_t, \quad 
 \sigma_t=\sigma_0 e^{Y_t}; \qquad dY_t= -\gamma Y_tdt+\theta dB^H_t,\label{FOU.eq}\ea
where $\gamma = \theta =\sigma_0=1$, $B$ is a Brownian motion and $B^H$   fractional Brownian motion with Hurst index $H\in(0,1)$. 
\end{example}
We compute the realized volatility and compare the estimated  roughness index $\widehat{ H}_{L,K}$ (with, $L=300\times300, K=300$) of instantaneous and realized volatility in the following table.
\[\begin{tabular}{|c|c|c|}\hline 
          H &  Instantaneous volatility  &  Realized volatility \\ \hline
         0.10 & 0.130 & 0.190 \\ \hline
         0.20 & 0.215 & 0.250 \\ \hline
         0.30 & 0.310 & 0.258 \\ \hline
         0.40 & 0.413 & 0.207 \\ \hline
         0.50 & 0.507 & 0.130 \\ \hline
         0.60 & 0.601 & 0.087 \\ \hline
         0.70 & 0.678 & 0.061 \\ \hline
         0.80 & 0.756 & 0.052 \\ \hline
\end{tabular}\]
The corresponding pictures for price process, realized volatility and the instantaneous volatility from Model \ref{FOU.eq} with Hurst index $H=\{0.05,0.1,0.2,0.3,0.4,0.5,0.6,0.7,0.8\}$ respectively are presented in Figure \ref{fig-OU.fractional.RV.IV}. Visually we can observe that for smaller $H$, the instantaneous volatility is rougher than realized volatility but as we increase $H$ the realized volatility shows significantly rougher behaviour than the instantaneous volatility. In Figure \ref{realizedH.fig} for the simulated models in Figure \ref{fig-OU.fractional.RV.IV}, we plot the estimated roughness index $\widehat{ H}_{L,K}(RV)$ and $\widehat{H}_{L, K}(\sigma)$ respectively in the red and blue line.
Though the estimated roughness index of instantaneous volatility (represented in blue line) gives an accurate estimate of Hurst index $H$, the  roughness index for realized volatility always stays below $0.3$. In particular, when the   instantaneous volatility exhibits smoother behaviour (corresponding to $H\geq 0.5$) the estimated roughness index of realized volatility turns out to be a   poor estimate for the  Hurst index. 

  Figure \ref{realizedH1.fig} shows the corresponding estimators $\widehat{ H}_{L,K}(RV)$ and $\widehat{H}_{L,K}(\sigma)$ for $100$ independent simulated price paths from \eqref{FOU.eq}.  The bold black lines represent the mean  across $100$ independent simulations whereas the dotted lines represent the corresponding $25\%$ and $75\%$ confidence intervals. For the price process  \eqref{FOU.eq}, no matter what the value of the Hurst exponent for instantaneous volatility, the roughness index of  realized volatility $\widehat{H}_{L,K}(RV)$ is always estimated to be between $0$ to $0.3$. 
  
  These examples illustrate our point:  one {\it cannot draw the conclusion} that (spot) `volatility is rough' i.e. reject the null hypothesis $H(\sigma)=1/2$ just because  realized volatility exhibits `rough' behaviour with $\widehat{H}_{L,K}(RV) <\frac{1}{2}$ or $\widehat{H}_R< 1/2$, {\it even} when these estimators exhibit values well below $1/2$.

\begin{figure}[hb!]
    \centering
    \includegraphics[scale=.57]{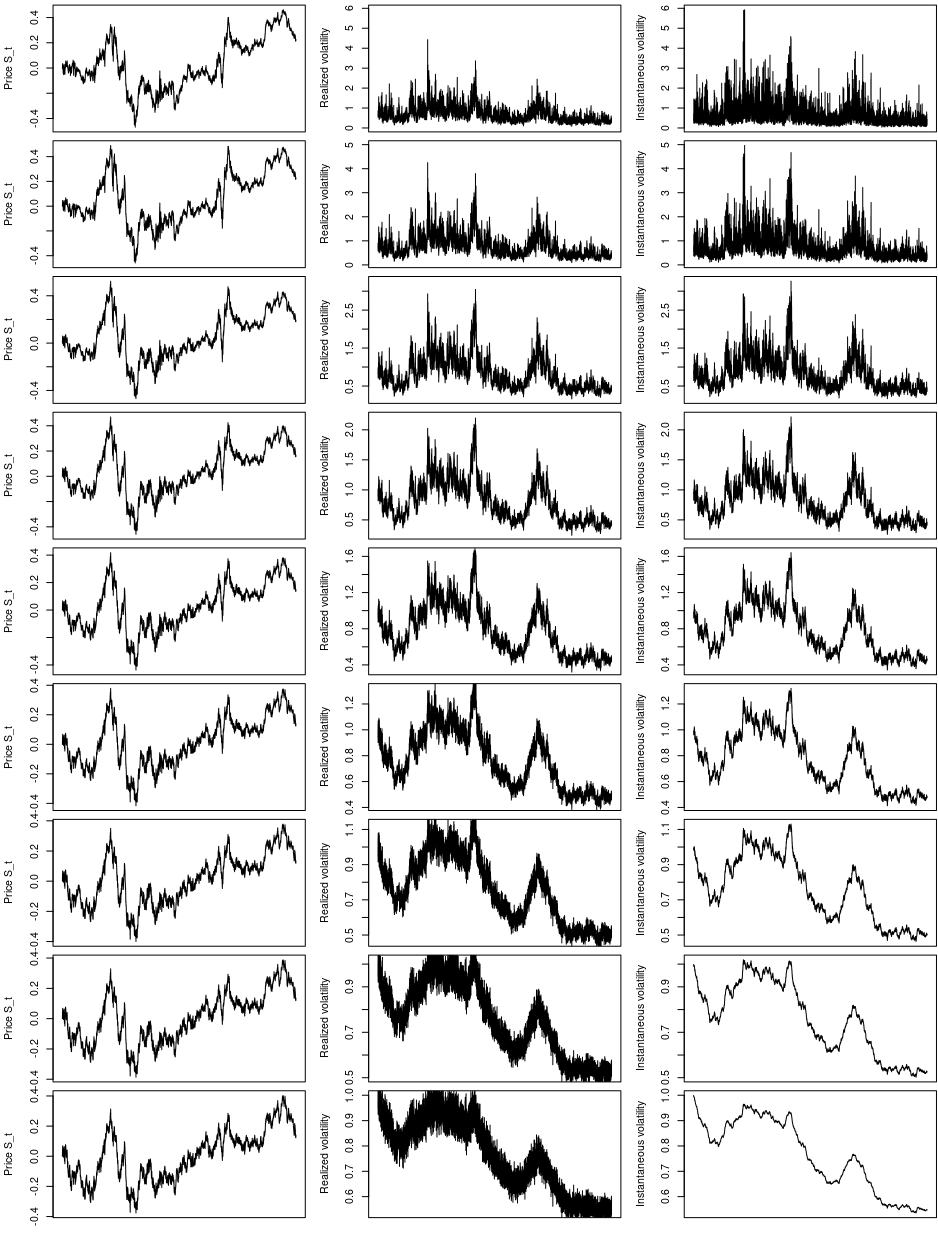}
    \caption{\textbf{Left: }Simulated price path $S_t$ of fractional OU model (Equation \ref{FOU.eq}) with\\ H=\{0.05,0.1,0.2,0.3,0.4,0.5,0.6,0.7,0.8\} respectively, \textbf{Center:} Realized volatility, \textbf{Right:} Instantaneous volatility.}
    \label{fig-OU.fractional.RV.IV}
\end{figure}

\begin{figure}[ht!]
    \centering
    \includegraphics[width=0.7\textwidth]{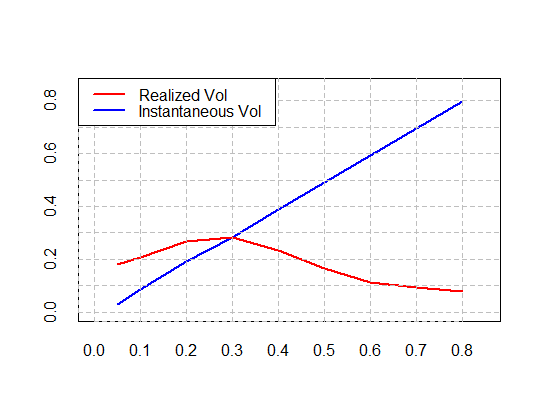}
    \caption{Estimated roughness index $\widehat{H}_{L=300\times 300,K=300}$ for realized volatility and instantaneous volatility from a high-frequency fractional-OU stochastic volatility model (Equation \eqref{FOU.eq}), plotted for price path generated with different values of $H$.}
    \label{realizedH.fig}
\end{figure}

\begin{figure}[ht!]
    \centering
    \includegraphics[width=0.7\textwidth]{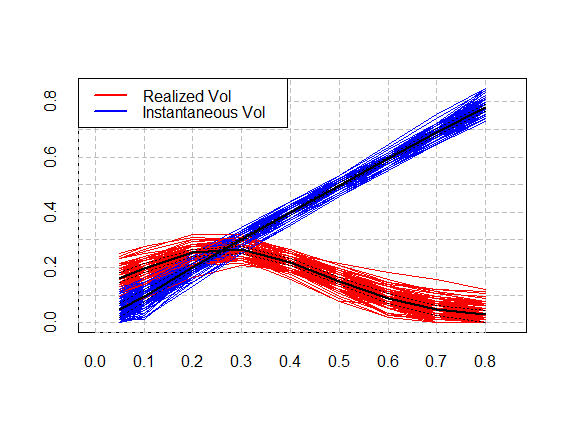}
    \caption{   Estimated roughness index $\widehat{H}_{L=300\times 300, K=300}$ for realized volatility and instantaneous volatility for a high-frequency simulation of a fractional-OU stochastic volatility model   \eqref{FOU.eq}, plotted against different Hurst index $H$ for $100$ independent price paths.}
    \label{realizedH1.fig}
\end{figure}
\clearpage

\section{Application to high-frequency financial data}
Having extensively explored the performance of our roughness estimator $\widehat{H}_{L, K}$ based on the normalized $p$-th variation statistic for spot and realized volatility on simulated price paths, we now apply it to  high-frequency financial time series. 

\subsection{AAPL} 
Figure \ref{apple.price} (left) shows the second-by-second record of AAPL stock price. The right graph of Figure \ref{apple.price} is $\log(W(L=1900\times 1900,K=1900,\pi,p,t=1,X))$ plotted against Hurst parameter $H=1/p$ for `AAPL' price. 
\par The left plot of Figure \ref{fig-apple.vol.estimate} represents $1$-minute realized volatility of   `AAPL' in $2016$. The right graph of Figure \ref{fig-apple.vol.estimate} is $\log(W(L=310\times310,K=310,\pi,p,t=1, X=RV))$ plotted against   $H=1/p$ for the 1-min AAPL realized volatility. Fixing the value of $L=310\times 310$, if we deviate the value of $K$ a little, then the estimated roughness index varies between $0.08$ to $0.22$. This is consistent with the results of  \cite{gatheral2014} regarding realized volatility. But as our simulation study suggests, the roughness index of realized volatility may be very different from that of spot  volatility which is the quantity modelled in continuous-time stochastic volatility models.
\subsection{SP500}
Several studies on rough volatility, including the original study \cite{gatheral2014}, are based on the  Oxford-Man Institute Realized Volatility dataset \footnote{\url{https://realized.oxford-man.ox.ac.uk/data}}.
Figure \ref{fig-oxford.man.realised.vol} represents the plot of  5-minute realized volatility of SP500. The X-axis represents date. The right graph of Figure \ref{fig-oxford.man.realised.vol} is $\log(W(L=70\times70,K=70,\pi,p,t=1,X=RV))$ plotted against Hurst parameter $H=1/p$ for the 5-min Oxford-Man institute realized volatility data. Fixing the value of $L=70\times 70$, if we deviate the value of $K$ a little, the estimated roughness index $\widehat H_{L, K}$ varies between $0.05$ to $0.25$. This finding is again consistent with the findings in  \cite{gatheral2014}. 

Overall, the picture that emerges from SP500 and AAPL data is quite similar to the one observed in simulations of diffusion-type stochastic volatility models discussed in Section \ref{sec.SV}.
As observed in Section \ref{sec.simulation}, these observations are fully compatible with a diffusion-type stochastic volatility model such as \eqref{eq.OU-SV} and one  cannot reject the null hypothesis $H(\sigma)=1/2$ just because  realized volatility exhibits `rough' behaviour with $\widehat{H}_{L,K}(RV) <\frac{1}{2}$ or $\widehat{H}_R< 1/2$, {\it even} though these estimators exhibit values around $0.1$.

\begin{figure}[ht!]
\centering
    \includegraphics[width=.9 \textwidth]{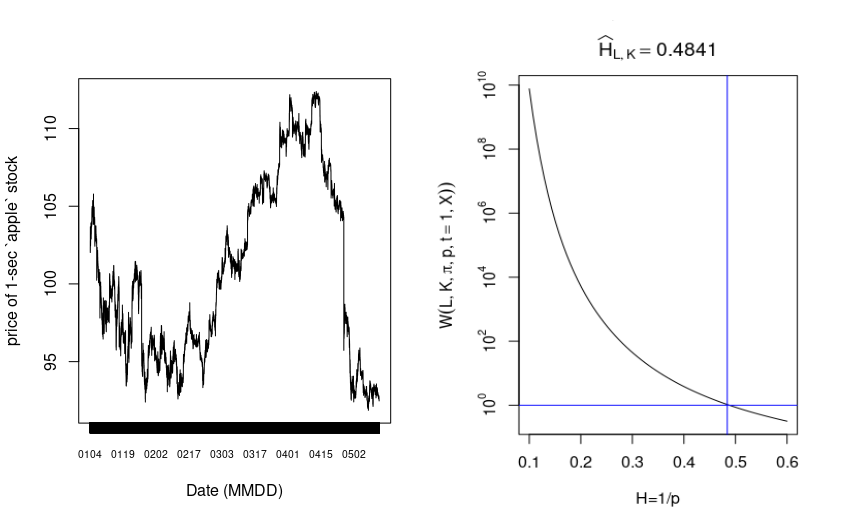}
     \caption{\textbf{Left: }   price of AAPL  04/Jan/2016 -  11/May/2016 (90 days). \textbf{Right: } Roughness index estimator $\widehat H_{L,K}$(with $L=1400\times 1400,K=1400$) for AAPL stock price data. }
    \label{apple.price}
\end{figure}

\begin{figure}[ht!]
    \centering
    \includegraphics[width=.9 \textwidth]{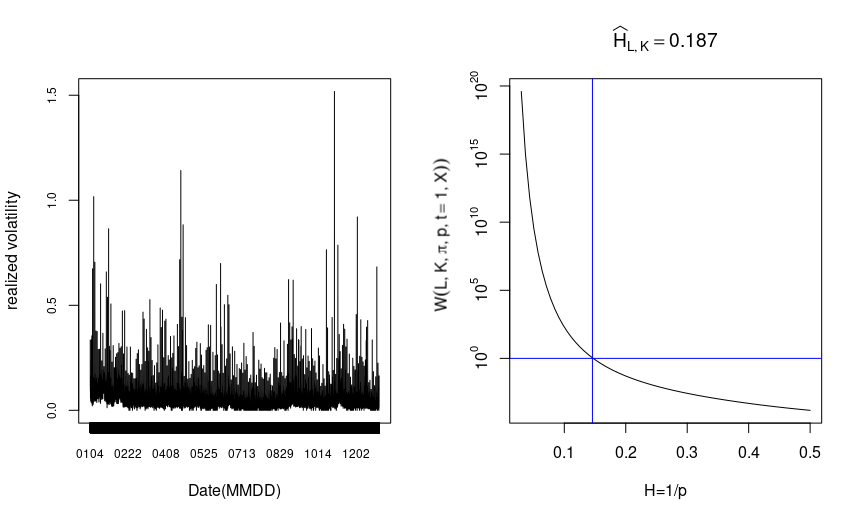}
    \caption{\textbf{Left: } Plot of 1-min realized volatility of `AAPL' (year 2016). \textbf{Right: }  Estimated roughness index $\widehat H_{L,K}$ (with $L=310\times 310,K=310$) for the 1-min realized volatility (estimated roughness index $\widehat H_{L,K}\in [.08-.22]$). }
    \label{fig-apple.vol.estimate}
\end{figure}

\clearpage
\begin{figure}[ht!]
    \centering
    \includegraphics[width=.9 \textwidth]{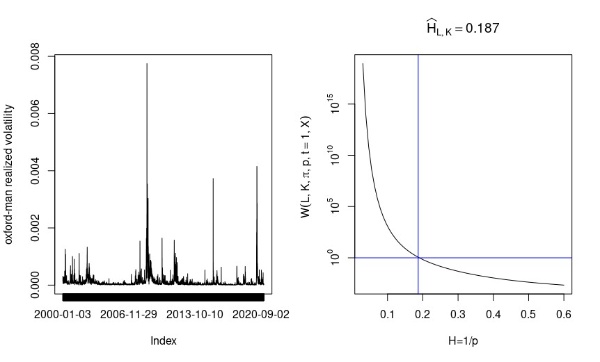}
    \caption{\textbf{Left: } S\&P500  5-min realized volatility. \textbf{Right: }  Estimated roughness index  $\widehat H_{L,K}\in [.05-.25]$ with $L=70\times70,K=70$. }
    \label{fig-oxford.man.realised.vol}
\end{figure}

\section{Rough volatility ... or estimation error?}
Given the  large literature on `rough volatility' in quantitative finance, it is somewhat surprising that the initial claim \cite{gatheral2014} that one needs to model the spot volatility process using a `rough' fractional noise with Hurst exponent $H< 1/2$ has not been examined more closely, especially given that several follow-up studies \citep{fukasawa2019,roger2019} point to the fact that the observations in \cite{gatheral2014} may well be compatible with a  Brownian diffusion model for volatility.

Our detailed examples illustrate that,
for stochastic-volatility diffusion models driven by Brownian motion as described in Examples \ref{example.simulation} and \ref{ex-OU-SV},  the realized volatility has a roughness index $\widehat{H_{L, K}}\approx 0.3$ so  exhibits an `apparent roughness' which instantaneous volatility does not have, both in terms of normalized $p$-th variation statistics and also in terms of the log- regression method used by  \cite{gatheral2014}. Clearly in these simulation examples this is entirely due to the discretization error or `estimation error'.

These results suggest that one cannot hastily conclude that the roughness observed in realized volatility is an indicator of similar behaviour in  spot volatility, as implicitly assumed in the `rough volatility' literature; the observations in high-frequency financial data are in fact compatible with a stochastic volatility model drive by Brownian motion and the origin of this apparent roughness may very well lie in estimation error rather than the noise process driving spot volatility. 
 
Also, as shown in Example \ref{ex-fou},  the rough behaviour of realized volatility does not lead us to reject the hypothesis that the underlying spot  volatility may be  modeled with a Brownian diffusion model or even a smoother model with long-range dependence and $H>1/2$. This observation, together with “Occam's razor", pleads for the use of  Markovian stochastic volatility models which seem compatible with the empirical evidence but are far more tractable. 

We are thus drawn to concur with \cite{roger2019} that ``the notion that volatility is rough, that is, governed by a fractional Brownian motion (with $H< 1/2$), is not an incontrovertible established fact; simpler models explain the observations just as well."


\bibliographystyle{apalike} 
\bibliography{RoughVolatility}       

\appendix
\subsection*{Proof of Theorem \ref{infinite.zero.weighted.variation} (i).}
We will prove the result for $t=T$, for general $t$ the proof generalizes without further complication. Since $q>p$ and the function $x$ has finite $p$-variation along $\pi$, the pathwise $q$-th variation $[x]^{(q)}_{\pi}(t)= 0$ for all $t\in [0,T]$. Hence given any fixed $h>0$ we have $[x]^{(q)}_{\pi}(t+h)-[x]^{(q)}_{\pi}(t) = 0$. Now given any $M>0$ and any $n>1$, we will show that $w(x,q,\pi)(t)>M$. 
For fix $n>1$ we have:
\[\text{ for all } t^n_i \in \pi^n, \qquad  [x]^{(q)}_{\pi}(t^n_{i+1})-[x]^{(q)}_{\pi}(t^n_i) = 0.\]
Define the `unfeasible estimator' 
\[w^n(x,q,\pi)(T) = \sum_{\pi^n}\frac{|x(t^n_{i+1})-x(t^n_{i})|^q}{[x]^{(q)}_{\pi}(t^n_{i+1})-[x]^{(q)}_{\pi}(t^n_{i})} \times (t^n_{i+1}-t^n_{i}).\]
Since $x\in V^{p}_\pi([0,T],\mathbb{R})$ for some $p>0$, for fix large enough $n$, we have $\sum_{\pi^n}|x(t^n_{i+1})-x(t^n_{i})|^q  \times (t^n_{i+1}-t^n_{i})>0$.
So we get the lower bound of $ w^n(x,q,\pi)$ as follows.
\[w^n(x,q,\pi)(T) \geq  \min_{t^n_i\in \pi^n} \left[\frac{1}{[x]^{(q)}_{\pi}(t^n_{i+1})-[x]^{(q)}_{\pi}(t^n_{i})} \right] \sum_{\pi^n}|x(t^n_{i+1})-x(t^n_{i})|^q \times (t^n_{i+1}-t^n_{i})\]
\[ \geq C \left[\frac{1}{\max_{t^n_i\in \pi^n}([x]^{(q)}_{\pi}(t^n_{i+1})-[x]^{(q)}_{\pi}(t^n_{i}))} \right] =\infty >M.\]
Since for all $n>1$, $w^n(x,q,\pi)(t) = \infty$ we can conclude the following: 
\[w(x,q,\pi)(t) = \lim_{n\to \infty} w^n(x,q,\pi)(t) = \infty.\]

\subsection*{Proof of Theorem \ref{infinite.zero.weighted.variation}(ii).}

We will prove the result for $t=T$, for general $t$ the proof follows exactly the same way. Since the function $x$ has finite $p$-variation along $\pi$ and $q<p$ from the assumption, the pathwise $q$ -th variation $[x]^{(q)}_{\pi}(t)= \infty$ for all $t\in [0,T]$. Hence, given any fixed $h>0$ we have $[x]^{(q)}_{\pi}(t+h)-[x]^{(q)}_{\pi}(t) = \infty$. Now given any $\epsilon>0$ and any $n>1$, we will show that $w(x,q,\pi)(t)<\epsilon$ by showing that $w^n(x,q,\pi)(t)<\epsilon$ for all large $n$. 
For fix $n>1$ we have:
\[\text{ for all } t^n_i \in \pi^n, \qquad  [x]^{(q)}_{\pi}(t^n_{i+1})-[x]^{(q)}_{\pi}(t^n_i) = \infty. \text{ So, }\]
\[w^n(x,q,\pi)(T) = \sum_{\pi^n}\frac{|x(t^n_{i+1})-x(t^n_{i})|^q}{[x]^{(q)}_{\pi}(t^n_{i+1})-[x]^{(q)}_{\pi}(t^n_{i})} \times (t^n_{i+1}-t^n_{i})\]\[\leq \max_{t^n_i\in \pi^n}\left[\frac{|x(t^n_{i+1})-x(t^n_{i})|^q}{[x]^{(q)}_{\pi}(t^n_{i+1})-[x]^{(q)}_{\pi}(t^n_{i})}\right]  \times \sum_{\pi^n}(t^n_{i+1}-t^n_{i})\] \[\leq T \times \max_{t^n_i\in \pi^n}\left[\frac{|x(t^n_{i+1})-x(t^n_{i})|^q}{[x]^{(q)}_{\pi}(t^n_{i+1})-[x]^{(q)}_{\pi}(t^n_{i})} \right]\]
\[ \leq T \times \max_{t^n_i\in \pi^n} |x(t^n_{i+1})-x(t^n_{i})|^q \times \left[\frac{1}{\min_{t^n_i\in \pi^n}([x]^{(q)}_{\pi}(t^n_{i+1})-[x]^{(q)}_{\pi}(t^n_{i}))} \right] =0 <\epsilon.\]
Since for all $n>1$, $w^n(x,q,\pi)(t) = 0$;
\[w(x,q,\pi)(t) = \lim_{n\to \infty} w^n(x,q,\pi)(t) = 0.\]

\subsection*{Proof of Theorem \ref{linear.weighted.variation}}
For convenience,  assume $g(u)= \frac{d}{du}[x]^{(p)}_\pi(u)$. Since the $p$-th variation is strictly increasing we have $\sup_{t\in [0,T]}\frac{1}{g(u)}$ $<\infty$. Since $g(u)$ is continuous in a compact interval $[0,T]$, it is also bounded. So for all $t\in [0,T]$, we have $\frac{1}{g(u)}\in (0,\infty)$. So as a consequence of mean value theorem,
\[ w^n(x,p,\pi)(t) = \sum_{\pi^n\cap [0,t]}\frac{|x(t^n_{i+1})-x(t^n_{i})|^p}{[x]^{(p)}_{\pi}(t^n_{i+1})-[x]^{(p)}_{\pi}(t^n_{i})} \times (t^n_{i+1}-t^n_{i})\]\[ =  \sum_{\pi^n\cap [0,t]}\frac{|x(t^n_{i+1})-x(t^n_{i})|^p}{g(u^n_i)}, \]
where  $u^n_i\in [t^n_i,t^n_{i+1}]$ for all $n\geq 1$ and for all $i=0,\cdots, N(\pi^n)-1$. Finally, using properties of Riemann integral we can conclude:
\[   \sum_{\pi^n\cap[0,t]}\frac{|x(t^n_{i+1})-x(t^n_{i})|^p}{g(u^n_i)} \xrightarrow[]{n\to \infty} \int_0^t \frac{1}{g(u)}d[x]^{(p)}_\pi(u)\]\[ = \int_0^t \frac{d[x]^{(p)}_\pi(u)/du }{g(u)}du = \int_0^tdu=t.\]
Since the limit always exists we can conclude the proof.

\subsection*{Proof of Example \ref{weightedQV-BM}}
Let $B$ be a Wiener process on the canonical Wiener space $(\Omega, {\cal F},\mathbb{P})$,   i.e $\Omega=C^0([0,T],\mathbb{R}), B(t,\omega)=\omega(t)$ and $\mathbb{P}$ is the Wiener measure.
Let $\pi^n=(0=t^{n}_0<t^n_1<\cdots<t^{n}_{N(\pi^n)}=T)$ be a sequence of partitions of $[0,T]$ satisfying $|\pi^n|\log n \to 0$. Then the results of Dudley \cite{dudley1973}   imply that
$$\mathbb{P}\left( \sum_{\pi^n} |B(t^n_{i+1}\wedge t)-B(t^n_{i}\wedge t)|^2\mathop{\to}^{n\to\infty} t\right)=1. $$
So if we set $\Omega_0=\Omega\cap Q_\pi([0,T],\mathbb{R})$ then
$\mathbb{P}\left(\Omega_0\right)=1$ and any $\omega\in \Omega_0$ satisfies $[\omega]_\pi(t)=t$. Now for any $\omega\in \Omega_0$ we also have the following:
\[w(\omega,2,\pi)(t) = \lim_{n\to \infty}w^n(\omega,2,\pi)(t) = \lim_{n\to \infty}\sum_{\pi^n\cap[0,t]}\frac{\left(\omega(t^n_{i+1})-\omega(t^n_{i})\right)^2}{[\omega]_{\pi}(t^n_{i+1})-[\omega]_{\pi}(t^n_{i})} \times (t^n_{i+1}-t^n_{i}).\]
\[ =\lim_{n\to \infty}\sum_{\pi^n\cap[0,t]}\frac{(\omega(t^n_{i+1})-\omega(t^n_{i}))^2}{(t^n_{i+1}-t^n_{i})} \times (t^n_{i+1}-t^n_{i}) =\lim_{n\to \infty} \sum_{\pi^n\cap[0,t]}(\omega(t^n_{i+1})-\omega(t^n_{i}))^2 =t . \]
So for Brownian motion, normalized quadratic variation   is almost surely equal to $t$. 
\begin{remark}
From Lemma \ref{weightedQV-BM} we know that, for any partition sequence $\pi$ with $|\pi^n|\log (n) \to 0$, there exists $\Omega_{\pi}\subset \Omega$ with $\mathbb{P}(\Omega_{\pi})=1$ such that:
\[\forall \omega\in \Omega_{\pi}, \; \forall t\in[0,T], \qquad w(\omega,2,\pi)(t) =t.\]
We also have the following relation between quadratic variation and normalized-QV for Brownian paths in the class $\Omega_\pi$.
\[\forall \omega\in \Omega_{\pi}, \; \forall t\in[0,T], \qquad w(\omega,2,\pi)(t) =[\omega]_{\pi} (t) .\]
i.e. the null set for quadratic variation and normalized-quadratic variation of Brownian motion are the same.
\end{remark}

\end{document}